\documentclass[draftclsnofoot,onecolumn]{IEEEtran}
\usepackage{amsfonts}
\usepackage{dsfont}
\usepackage{amssymb}
\usepackage{graphicx}
\usepackage{subfigure}
\usepackage{enumerate}
\usepackage{amsmath}
\usepackage{color}
\usepackage{amsthm}
\usepackage{amsmath}
\usepackage{algorithm}
\usepackage{algpseudocode}
\usepackage[bookmarks=false]{hyperref}
\hypersetup{hidelinks}
\usepackage{breakurl}
\usepackage{cite}
\usepackage{epstopdf}

\newcommand{\bp}{\begin{proof} \small }
\newcommand{\ep}{\end{proof} \normalsize}
\newcommand{\epx}{\end{proof} \small}
\newcommand{\bpa}{\begin{proofappx} \footnotesize }
\newcommand{\epa}{\end{proofappx} \small }
\newtheorem{theorem}{Theorem}
\newtheorem{proposition}{Proposition}

\newtheorem{lemma}{Lemma}

\newtheorem*{theorem*}{Theorem}
\newtheorem*{proposition*}{Proposition}
\newtheorem*{corollary*}{Corollary}
\newtheorem*{lemma*}{Lemma}
\newtheorem*{assumption*}{Assumption}
\newtheorem*{definition*}{Definition}
\newtheorem*{claim*}{Claim}

\newcommand{\be}{\begin{equation}}
\newcommand{\ee}{\end{equation}}
\newcommand{\bs}{\begin{subequations}}
\newcommand{\es}{\end{subequations}}
\newcommand{\bq}{\begin{eqnarray}}
\newcommand{\eq}{\end{eqnarray}}
\newcommand{\bqn}{\begin{eqnarray*}}
\newcommand{\eqn}{\end{eqnarray*}}

\newcommand{\ba}{\left[ \begin{array}}
\newcommand{\ea}{\\ \end{array} \right]}
\newcommand{\ben}{\begin{enumerate}}
\newcommand{\een}{\end{enumerate}}

\def\A{{\boldsymbol{A}}}

\def\H{{\boldsymbol{H}}}

\def\l{{\boldsymbol{l}}}

\def\w{{\boldsymbol{w}}}
\def\x{{\boldsymbol{x}}}

\def\z{{\boldsymbol{z}}}

\def\real{{\mathchoice%
{\hbox{\rm\setbox1=\hbox{I}\copy1\kern-.45\wd1 R}}
{\hbox{\rm\setbox1=\hbox{I}\copy1\kern-.45\wd1 R}}
{\hbox{\scriptsize\rm\setbox1=\hbox{I}\copy1\kern-.45\wd1 R}}
{\hbox{\scriptsize\rm\setbox1=\hbox{I}\copy1\kern-.45\wd1 R}}}}

\def\Zint{{\mathchoice{\setbox1=\hbox{\sf Z}\copy1\kern-.75\wd1\box1}
{\setbox1=\hbox{\sf Z}\copy1\kern-.75\wd1\box1}
{\setbox1=\hbox{\scriptsize\sf Z}\copy1\kern-.75\wd1\box1}
{\setbox1=\hbox{\scriptsize\sf Z}\copy1\kern-.75\wd1\box1}}}
\newcommand{\complex}{ \hbox{\rm C\kern-0.45em\rule[.07em]{.02em}{.58em}%
\kern 0.43em}}

\begin{document}
	%
\title{Heterogeneous Coded Computation across 
	Heterogeneous Workers}
	
	\author{\IEEEauthorblockN{Yuxuan Sun$^*$, Junlin Zhao$^\dagger$, Sheng Zhou$^*$, Deniz G\"und\"uz$^\dagger$}\\
		\IEEEauthorblockA{$^*$Beijing National Research Center for Information Science and Technology\\
			Department of Electronic Engineering, Tsinghua University, Beijing 100084, China\\
        $^\dagger$Department of Electrical and Electronic Engineering, Imperial College London, London SW7 2BT, UK\\
        Email: \{sunyx15@mails., sheng.zhou@\}tsinghua.edu.cn, \{j.zhao15, d.gunduz\}@imperial.ac.uk}
    \thanks{This work is sponsored in part by the European Research Council (ERC) under Starting Grant BEACON (grant No. 725731), the Nature Science Foundation of China (No. 61871254, No. 91638204, No. 61571265, No. 61861136003, No. 61621091), National Key R\&D Program of China 2018YFB0105005, and Intel Collaborative Research Institute for Intelligent and Automated Connected Vehicles.}}

	\maketitle

	\begin{abstract}
		Coded distributed computing framework enables large-scale machine learning (ML) models to be trained efficiently in a distributed manner, while mitigating the straggler effect.
		In this work, we consider a multi-task assignment problem in a coded distributed computing system, where multiple masters, each with a different matrix multiplication task, assign computation tasks to workers with heterogeneous computing capabilities.
		Both \emph{dedicated} and \emph{probabilistic} worker assignment models are considered, with the objective of minimizing the average completion time of all computations.
		For dedicated worker assignment, greedy algorithms are proposed and the corresponding optimal load allocation is derived based on the Lagrange multiplier method.
		For probabilistic assignment, successive convex approximation method is used to solve the non-convex optimization problem.
		Simulation results show that the proposed algorithms reduce the completion time by $80\%$ over uncoded scheme, and $49\%$ over an unbalanced coded scheme.

	\end{abstract}
%
		
	%
	\IEEEpeerreviewmaketitle
	
\section{Introduction}

Machine learning (ML) techniques are penetrating into many aspects of human lives, and boosting the development of new applications from autonomous driving, virtual and augmented reality, to Internet of things \cite{park2018wireless}.
Training complicated ML models requires computations with massive volumes of data, e.g., large-scale matrix-vector multiplications, which cannot be realized on a single centralized computing server.
Distributed computing frameworks such as MapReduce \cite{lee2018speeding} enable a centralized \emph{master} node to allocate data and update global model, while tens or hundreds of distributed computing nodes, called \emph{workers}, train ML models in parallel using partial data. Since task completion time depends on the slowest worker, a key bottleneck in distributed computing is the \emph{straggler effect}: experiments on Amazon EC2 instances show that some workers can be 5 times slower than the typical performance \cite{tandon2017}. 

Straggler effect can be mitigated by adding redundancy to the distributed computing system via coding \cite{tandon2017, li2016unified, lee2018speeding, hierarchical,codedhet, nonpersistent,new_dutta2018}, or by scheduling computation tasks \cite{li_nearoptimal,new_on_the_effect,mma2019computation}.
Maximum distance separable (MDS) codes are widely applied for matrix multiplications \cite{tandon2017, li2016unified, lee2018speeding, hierarchical,codedhet, new_dutta2018}, which can reduce the task completion time by $O(\log N)$, where $N$ is the number of workers \cite{lee2018speeding}.
A unified coded computing framework for straggler mitigation is proposed in \cite{li2016unified}.
Heterogeneous workers are considered in \cite{codedhet}, and an asymptotically optimal load allocation scheme is proposed.
Although the stragglers are slower than the typical workers, they can still make non-negligible contributions to the system \cite{hierarchical, nonpersistent}. A hierarchical coded computing framework is thus proposed in \cite{hierarchical}, where tasks are partitioned into multiple levels so that stragglers contribute to subtasks in the lower levels.
Multi-message communication with Lagrange coded computing is used in \cite{nonpersistent} to exploit straggler servers.

The straggler effect can be mitigated even with uncoded computing, via redundant scheduling of tasks and multi-message communications.
A batched coupon's collector scheme is proposed in \cite{li_nearoptimal}, and the expected completion time is analyzed in \cite{new_on_the_effect}. The input data is partitioned into batches, and each worker randomly processes one at a time, until the master collects all the results.
Deterministic scheduling orders of tasks at different workers are proposed in  \cite{mma2019computation}, specifically cyclic and staircase scheduling, and the relation between redundancy and task completion time is characterized.

Existing papers mainly consider a single master.
However, in practice, workers may be shared by more than one masters to carry out multiple large-scale computation tasks in parallel. 
Therefore, in this work, we focus on a multi-task assignment problem for a heterogeneous distributed computing system using MDS codes. As shown in Fig. \ref{system}, we consider multiple masters, each with a matrix-vector multiplication task, and a number of workers with \emph{heterogeneous} computing capabilities. The goal is to design centralized worker assignment and load allocation algorithms that minimize the completion time of all the tasks. We consider both \emph{dedicated} and \emph{probabilistic} worker assignment policies, and formulate a non-convex optimization problem under a unified framework.
For dedicated assignment, each worker serves one master.
The optimal load allocation is derived, and the worker assignment is transformed into a max-min allocation problem, for which NP-hardness is proved and greedy algorithms are proposed.
For probabilistic assignment, each worker selects a master to serve based on an optimized probability, and a successive convex approximation (SCA) based algorithm is proposed.
Simulation results show that 
the proposed algorithms can drastically reduce the task completion time compared to uncoded and unbalanced coded schemes.

The rest of the paper is organized as follows. The system model and problem formulation is introduced in Sec. \ref{sys}.
Dedicated and probabilistic worker assignments, and the corresponding load allocation algorithms are proposed in Sec. \ref{dedicated} and Sec. \ref{oppor}, respectively. 
Simulation results are presented in Sec. \ref{sim}, and the conclusions are summarized in Sec. \ref{con}.

\begin{figure}[!t]
	\centering
	\includegraphics[width=0.6\textwidth]{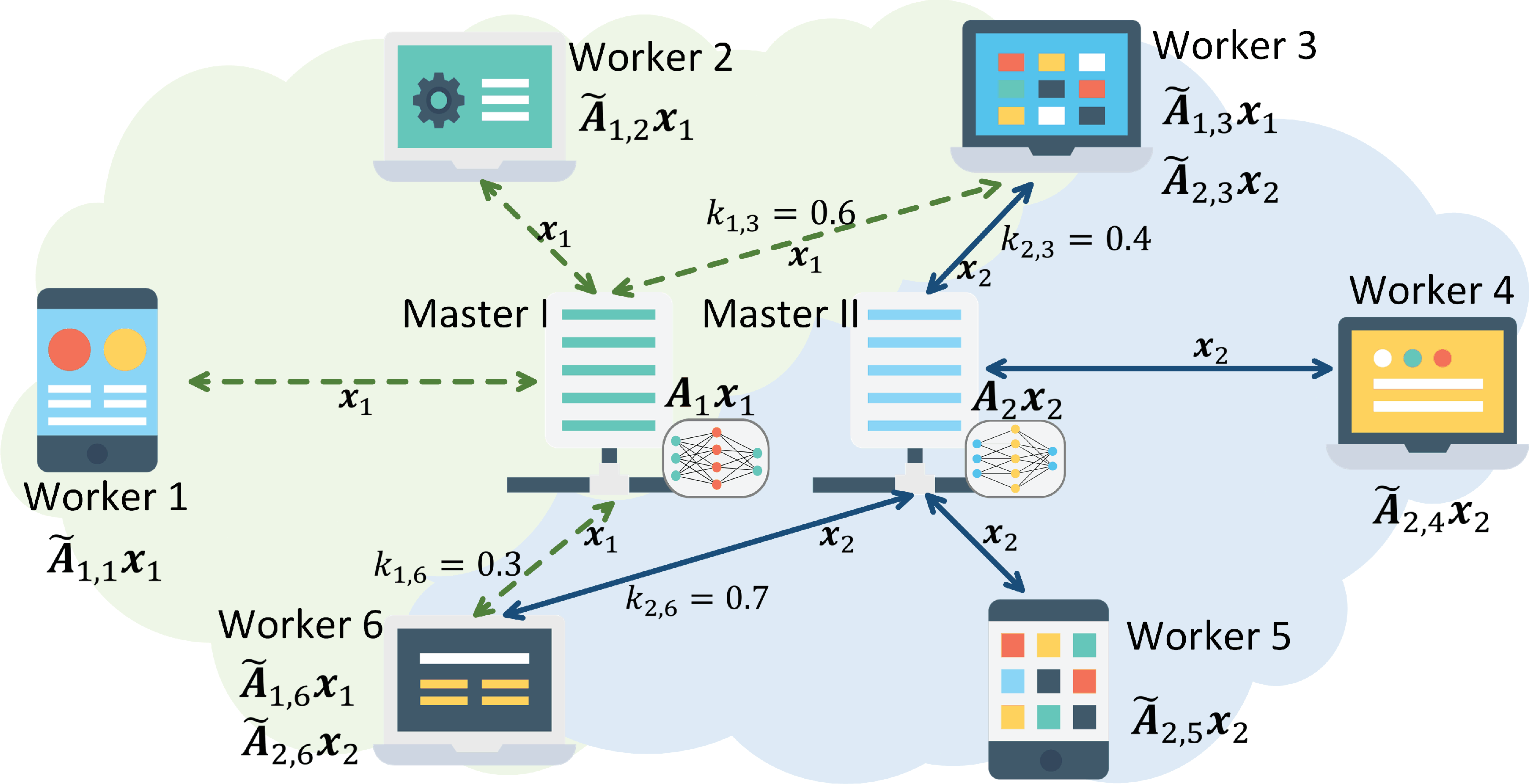}
	\caption{Illustration of a distributed computing system with multiple master and worker nodes.}	\label{system}
	\vspace{-3mm}
\end{figure}

\newtheorem{remark}{Remark}
\section{System Model and Problem Formulation}  \label{sys}

\subsection{System Architecture}
We consider a heterogeneous distributed computing system with $M$ masters $\mathcal{M}=\{1,2,...,M\}$, and $N$ workers $\mathcal{N}=\{1,2,...,N\}$, with $N>M$. 
We assume that each master has a matrix-vector multiplication task\footnote{In training ML models, e.g., linear regression, matrix-vector multiplication tasks are carried out at each iteration of the gradient descent algorithm. These tasks are independent over iterations, thus we focus on one iteration here.}.
The task of master $m$ is denoted by $\A_m\x_m$, where $\A_m\in \mathbb{R}^{L_m\times s_m}$, $\x_m\in \mathbb{R}^{s_m}$, $L_m,s_m \in \mathbb{Z}^+$. 
The masters can use the workers to complete their computation tasks in a distributed manner.


To deal with the straggling workers, we adopt MDS coded computation, and encode the rows of $\A_m$. Define the coded version of $\A_m$ as $\boldsymbol{\tilde{A}}_m$, which is further divided into $N$ sub-matrices:
\begin{align}
\boldsymbol{\tilde{A}}_m=\left[\boldsymbol{\tilde{A}}_{m,1}^T,~\boldsymbol{\tilde{A}}_{m,2}^T,~\cdots,~ \boldsymbol{\tilde{A}}_{m,N}^T\right]^T,
\end{align}
where $\boldsymbol{\tilde{A}}_{m,n}\in \mathbb{R}^{l_{m,n}\times s_m}$ is assigned to worker $n$, and $l_{m,n}$ is a non-negative integer representing the load allocated to worker $n$. 
Vector $\x_m$ is multicast from master $m$ to the workers with $l_{m,n}>0$, and worker $n$ calculates the multiplication of $l_{m,n}$ coded rows of $\A_m$ (which is $\boldsymbol{\tilde{A}}_{m,n}$) and $\x_m$.
Matrix $\A_m$ is thus $\left(\sum_{n=1}^{N}l_{m,n}, L_m\right)$-MDS-coded, with the requirement of $\sum_{n=1}^{N}l_{m,n}\geq L_m$.
Upon aggregating the multiplication results for any $L_m$ coded rows of $\A_m$, master $m$ can recover $\A_m\x_m$.

\subsection{Task Processing Time}
The processing times of the assigned computation tasks at the workers are modeled as mutually independent random variables. Following the literature on coded computing \cite{ li2016unified, lee2018speeding, hierarchical,codedhet}, the processing time at each worker is modeled by a shifted exponential distribution\footnote{In this work, the worker assignment and load allocation algorithms are designed based on the assumption of shifted exponential distribution. However, the proposed methods can also be applied to other distributions, as long as the corresponding function $f(x,t)$ defined in \eqref{fxt} is convex.}. 
The processing time, $T_{m,n}^{[l_{m,n}]}$, for worker $n$ to calculate the multiplication of $l_{m,n}>0$ coded rows of $\A_m$ and $\x_m$ has the cumulative distribution function: 
\begin{align}
\mathbb{P}\left[T_{m,n}^{[l_{m,n}]}\leq t\right]=
\begin{cases}
&1-e^{-\frac{u_{m,n}}{l_{m,n}}\left(t-a_{m,n}l_{m,n}\right)},t\geq a_{m,n}l_{m,n},\\
&0, ~~ \text{otherwise},
\end{cases}
\end{align}
where $a_{m,n}>0$ is a parameter indicating the minimum processing time for one coded row, and $u_{m,n}>0$ is the parameter modeling the straggling effect.


We consider a \emph{heterogeneous environment} by assuming that $u_{m,n}$ and $a_{m,n}$ are different over different master-worker pairs $(m,n)$, for $\forall m \in \mathcal{M}$ and $\forall n \in \mathcal{N}$.
This assumption is due to the fact that workers may have different computation speeds, and the dimensions of $\A_m$ and $\x_m$ vary over $m$.


\subsection{Worker Assignment Policy}
We consider two worker assignment policies: 

	1) \emph{Dedicated worker assignment:}
	In this policy, each worker $n$ is assigned computation tasks from a single master $m\in \mathcal{M}$. Let indicator $k_{m,n}=1$ if worker $n$ provides computing service for master $m$, and $k_{m,n}=0$ otherwise.
	Since a worker serves at most one master, we have $\sum_{m=1}^{M}k_{m,n}\leq1, \forall n \in\mathcal{N}$.
	
	2) \emph{Probabilistic worker assignment:}
	In this policy, each worker randomly selects which master to serve according to probability $k_{m,n}\in [0,1]$.
	For each worker $n \in\mathcal{N}$, we have $\sum_{m=1}^{M}k_{m,n}\leq1$.
	In Fig. \ref{system}, worker $3$ selects master $1$ to serve with probability $0.6$, and master $2$ with probability $0.4$. 
	


\subsection{Problem Formulation} \label{for}
Let $X_{m,n}(t)$ denote the number of multiplication results (one result refers to the multiplication of one coded row of $\A_m$ with $\x_m$) master $m$ collects from worker $n$
till time $t$.
We assume that worker $n$ computes $\boldsymbol{\tilde{A}}_{m,n}\x_m$ and then sends the result to the master $m$ upon completion, without further dividing it into subtasks or transmitting any feedbacks before completion.
Therefore, master $m$ can either receive $l_{m,n}$ results or none from worker $n$ till time $t$. 
We denote the number of aggregated results at master $m$ until time $t$ by $X_m(t)$, and we have $X_{m}(t)=\sum_{n=1}^{N}X_{m,n}(t)$.

Our objective is to minimize the average completion time $t$, upon which all the masters can aggregate sufficient results from the workers to recover their computations with high probability. 
We aim to design a centralized policy that optimizes worker assignment $\{k_{m,n}\}$ and load allocation $\{l_{m,n}\}$.
The optimization problem is formulated as:
\begin{subequations}
\begin{align}
\mathcal{P}1: \min_{\{l_{m,n},k_{m,n},t\}} &~~~~~~t \label{ori_obj} \\
\text{s.t.} ~~~~~&\mathbb{P}\left[X_m(t)\geq L_m\right]\geq \rho_s, ~ \forall m,  \label{ori_cons_exp} \\
&\sum_{m=1}^{M} k_{m,n}\leq1, ~~\forall n,   \label{ori_cons_sum} \\
&k_{m,n}\in \mathcal{K}, ~~  l_{m,n} \in \mathbb{N},~~\forall m,n,   \label{ori_cons_k} 
\end{align}
\end{subequations}
where we have $\mathcal{K}=\{0,1\}$ for dedicated worker assignment, while $\mathcal{K}=[0,1]$ for probabilistic worker assignment, and $\mathbb{N}$ is the set of non-negative integers.
In constraint \eqref{ori_cons_exp}, $\rho_s$ is defined as the probability that master $m$ receives no less than $L_m$ results until time $t$; that is, the probability of $\A_m\x_m$ being recovered.
Constraint \eqref{ori_cons_sum} guarantees that under dedicated assignment, each worker serves at most one master, and under probabilistic assignment, the total probability rule is satisfied. 

The key challenge to solve $\mathcal{P}1$ is that, constraint \eqref{ori_cons_exp} cannot be explicitly expressed, since it is difficult to find all the combinations that satisfy $X_m(t)\geq L_m$ in a heterogeneous environment with non-uniform loads $\{l_{m,n}\}$.
Therefore, we instead consider an approximation to this problem, by substituting constraint \eqref{ori_cons_exp} with an expectation constraint:
\begin{subequations}
\begin{align}
\mathcal{P}2: \min_{\{l_{m,n},k_{m,n},t\}} &~~~~~~t \label{p1_obj} \\
\text{s.t.} ~~~~~&L_m-\mathbb{E}[X_m(t)]\leq 0, ~ \forall m,  \label{cons_exp} \\
&\text{Constraints } \eqref{ori_cons_sum}, \eqref{ori_cons_k}, \nonumber
\end{align}
\end{subequations}
where constraint \eqref{cons_exp} states that the expected number of results master $m$ receives until time $t$ is no less than $L_m$. A similar approach is used in \cite{codedhet}, where the gap between the solutions of $\mathcal{P}1$ and $\mathcal{P}2$ is proved to be bounded when there is a single master.
We will design algorithms that solve $\mathcal{P}2$ in the following two sections. 

Constraint \eqref{cons_exp} can be explicitly expressed. Let $\mathbb{I}_{\{x\}}$ be an indicator function with value $1$ if event $\{x\}$ is true, and $0$ otherwise.
If $k_{m,n}>0$ (and thus $l_{m,n}>0$),
\begin{align}
&\mathbb{E}[X_{m,n}(t)]=\mathbb{E}\left[k_{m,n}l_{m,n} \mathbb{I}_{\left\{T_{m,n}^{[l_{m,n}]}\leq t \right\}}\right] 
=
\begin{cases}
k_{m,n}l_{m,n}\left[1-e^{-\frac{u_{m,n}}{l_{m,n}}\left(t-a_{m,n}l_{m,n}\right)}\right],t\geq a_{m,n}l_{m,n},\\
0, ~~ \text{otherwise}.  
\end{cases} \label{x_mn_def}
\end{align}
If $k_{m,n}=0$ (and thus $l_{m,n}=0$), $\mathbb{E}[X_{m,n}(t)]=0$.
And we have $\mathbb{E}[X_{m}(t)]=\sum_{n=1}^{N}\mathbb{E}[X_{m,n}(t)]$.

%


The following observations help us simplify $\mathcal{P}2$:

1) From constraint \eqref{cons_exp}, we can infer that for $\forall m\in\mathcal{M}$, the optimal task completion time $t^*$ satisfies $t^*\geq \max_{\{n\in\Omega_m\}}\{a_{m,n}l_{m,n}\}$, where $\Omega_m\subset \mathcal{N}$ is the subset of workers serving master $m$.
In fact, if there exists $n_0\in\mathcal{N}$ such that $t^*< a_{m,n_0}l_{m,n_0}$, we have $\mathbb{E}[X_{m,n_0}(t^*)]=0$, i.e., master $m$ cannot expect to receive any results from worker $n_0$. By reducing $l_{m,n_0}$ to satisfy $\mathbb{E}[X_{m,n_0}(t^*)]>0$, it is possible to further reduce $t^*$.

2) Due to the high dimension of input matrix $\A_m$, $l_{m,n}$ is usually in the order of hundreds or thousands. So we relax the constraint $l_{m,n} \in \mathbb{N}$ to $l_{m,n}\geq0$, and omit the effect of rounding in the following derivations. 


%

Based on the two statements, by substituting \eqref{x_mn_def}, we simplify constraint \eqref{cons_exp} as: 
\begin{align}
L_m-\sum_{n=1}^{N}k_{m,n}l_{m,n}
\left(1-e^{-\frac{u_{m,n}}{l_{m,n}}\left(t-a_{m,n}l_{m,n}\right)}\right)\leq 0. \label{cons_exp_simple}
\end{align}
And problem $\mathcal{P}2$ can be simplified as follows:
\begin{subequations}
\begin{align}
\mathcal{P}3: \min_{\{l_{m,n},k_{m,n},t\}} &~~~~~~t \label{p2_obj} \\
\text{s.t.} ~~~~~& \text{Constraints } \eqref{ori_cons_sum}, \eqref{cons_exp_simple}, \nonumber\\
&k_{m,n}\in \mathcal{K}, ~~ l_{m,n} \geq 0, ~~\forall m,n.  \label{cons_l}
\end{align}
\end{subequations}


Problem $\mathcal{P}3$ is a non-convex optimization problem due to the non-convexity of \eqref{cons_exp_simple}, which is in general difficult to solve. In the following two sections, we will propose algorithms for dedicated and probabilistic worker assignments and corresponding load allocations, respectively. 

\section{Dedicated Worker Assignment} \label{dedicated}
In this section, we solve $\mathcal{P}3$ for dedicated worker assignment, where $\mathcal{K}=\{0,1\}$.
Given the assignment of workers, we first derive the optimal load allocation.
Then the worker assignment can be transformed into a max-min allocation problem, for which NP-hardness is shown and two greedy algorithms are developed.

\subsection{Optimal Load Allocation for a Given Worker Assignment}
We first assume that the subset of workers that serve master $m$ is given by $\Omega_m\subset \mathcal{N}$, and derive the optimal load allocation for master $m$, that minimizes the approximate completion time. 
The problem is formulated as:
\begin{subequations}
\begin{align}
\mathcal{P}4:~\min_{\{l_{m,n},~t_m\}} &~~~t_m \label{p3_obj} \\
\text{s.t.} ~~~~&L_m-\mathbb{E}[X_m(t_m)]\leq 0, \label{p3_cons1} \\ 
&l_{m,n} \geq 0, \forall n \in \Omega_m,  \label{p3_cons2}
\end{align}
\end{subequations}
where $t_m$ is defined as the approximate completion time of master $m$, $X_m(t_m)=\sum_{n\in \Omega_m}X_{m,n}(t_m)$ is the number of results aggregated at master $m$ till time $t_m$, and
\begin{align}
\mathbb{E}[X_m(t_m)]  
=\sum_{n\in \Omega_m}l_{m,n} \left( 1-e^{-\frac{u_{m,n}}{l_{m,n}}\left(t_m-a_{m,n}l_{m,n}\right)} \right).
\end{align}
\begin{lemma}
	Problem $\mathcal{P}4$ is a convex optimization problem.
\end{lemma}
\begin{proof}
	See Appendix \ref{proof1}.
\end{proof}

Let $\l_m\triangleq\{l_{m,n}\mid n\in \Omega_m\}$. The partial Lagrangian of  $\mathcal{P}4$ is given by
\begin{align}
	&\mathcal{L}(\l_m,t_m, \lambda_m)= t_m+\lambda_m\left(L_m-\mathbb{E}[X_m(t_m)]\right)=t_m+
	\lambda_m\left[L_m-\sum_{n\in \Omega_m}l_{m,n} \left( 1-e^{-\frac{u_{m,n}}{l_{m,n}}\left(t_m-a_{m,n}l_{m,n}\right)} \right) \right], \label{p3_lagrange}
\end{align}
where $\lambda_m\geq 0$ is the Lagrange multiplier associated with \eqref{p3_cons1}.

The partial derivatives of $\mathcal{L}$ can be derived as
\begin{align}
\frac{\partial \mathcal{L}}{\partial l_{m,n}}=
\lambda_m\left[\left(1+\frac{u_{m,n}t_m}{l_{m,n}}\right)e^{-\frac{u_{m,n}}{l_{m,n}}\left(t_m-a_{m,n}l_{m,n}\right)}-1 \right],  \label{partial_l}
\end{align}
\begin{align}
\frac{\partial \mathcal{L}}{\partial t_m}=1-\lambda_m\sum_{n\in \Omega_m}u_{m,n}e^{-\frac{u_{m,n}}{l_{m,n}}\left(t_m-a_{m,n}l_{m,n}\right)}.  \label{partial_t}
\end{align}

The optimal solution $(\l^*_m,t^*_m, \lambda^*_m)$ needs to satisfy the Karush-Kuhn-Tucker (KKT) conditions
\begin{subequations}
\begin{align}
&\frac{\partial \mathcal{L}}{\partial l^*_{m,n}}=0, ~\forall n\in \Omega_m ,
~~\frac{\partial \mathcal{L}}{\partial t^*_m}=0  \label{kkt_2} \\ 
&\lambda^*_m\left(L_m-\mathbb{E}[X_m(t^*_m)]\right)=0 \label{kkt_3}  \\
&\lambda^*_m\geq 0, ~l^*_{m,n}>0  \label{kkt_4}
\end{align}
\end{subequations}

Define $\mathcal{W}_{-1}(x)$ as the lower branch of Lambert W function, where $x\leq -1$ and $\mathcal{W}_{-1}(xe^x)=x$. Let 
\begin{align}\label{phi_def}
\phi_{m,n}\triangleq \frac{1}{u_{m,n}}\left[-\mathcal{W}_{-1}(-e^{-u_{m,n}a_{m,n}-1} )-1\right]. 
\end{align}
By solving KKT conditions \eqref{kkt_2}-\eqref{kkt_4}, the optimal load allocation for each individual master is given as follows.

\begin{theorem}
	For master $m\in\mathcal{M}$, and a given subset of workers $\Omega_m \in\mathcal{N}$ serving this master, the optimal load allocation $l^*_{m,n}$ derived from $\mathcal{P}4$, and the corresponding minimum approximate completion time $t_m^*$ are given by:
	\begin{align}
		&l^*_{m,n}=\frac{L_m}{\phi_{m,n}\sum_{n\in \Omega_m} \frac{u_{m,n}}{1+u_{m,n}\phi_{m,n}}},  \label{opt_l}\\
		&t_m^*= \frac{L_m}{\sum_{n\in \Omega_m} \frac{u_{m,n}}{1+u_{m,n}\phi_{m,n}}}. \label{opt_t}
	\end{align}
\end{theorem}
\begin{proof}
	See Appendix \ref{proof2}.
\end{proof}


\subsection{Greedy Worker Assignment Algorithms}
Now we consider how to assign workers to different masters to minimize the task completion time $t$. 	Let
\begin{align}
v_{m,n}\triangleq \frac{u_{m,n}}{L_m(1+u_{m,n}\phi_{m,n})}.  \label{vmn}
\end{align}
Based on Theorem 1, the worker assignment problem can be transformed into a max-min allocation problem, given in the following proposition.

\begin{proposition}
	Problem $\mathcal{P}3$ is equivalent to  
	\begin{subequations}
	\begin{align}
	\mathcal{P}5:  \max_{\{k_{m,n}\}}& \min_{m\in\mathcal{M}}\sum_{n=1}^{N}k_{m,n} v_{m,n} \\
	{\rm s.t.} ~& \sum_{m=1}^{M} k_{m,n}\leq 1, 
	~k_{m,n}\in\{0,1\},~ \forall m,n.
	\end{align}
\end{subequations}
\end{proposition}
\begin{proof}
	We use $t_m^*$ to represent the minimum task completion time of each master $m$ given the set of workers $\Omega_m$, and define $V_m\triangleq \frac{1}{t_m^*}$. From Theorem 1, we have:
	\begin{align}
	V_m=\frac{1}{L_m}\sum_{n\in \Omega_m} \frac{u_{m,n}}{1+u_{m,n}\phi_{m,n}} =\sum_{n=1}^{N}k_{m,n} v_{m,n}.  
	\end{align}
	Note that in $\mathcal{P}3$, 
	$t^*=\max_{m\in\mathcal{M}} t_m^*$. With $t_m^*>0$ and $V_m>0$, $\min_{\{k_{m,n}\}}\max_{m\in\mathcal{M}} t_m^*$ is equivalent to $ \max_{\{k_{m,n}\}} \min_{m\in\mathcal{M}}V_m$.
\end{proof}

\begin{algorithm} [!t]
	\caption{Iterated Greedy Algorithm for Dedicated Worker Assignment}
	\begin{algorithmic}[1]
		\State \textbf{Input}: $\Omega_m=\emptyset$, $V_m=0$, and $\{v_{m,n}\}$ according to \eqref{vmn}.
		\For {$n=1,...,N$}  \Comment{\textit{Initialization}}
		\State $m^*=\arg\max_{m\in \mathcal{M}} v_{m,n}$.
		\State $V_{m^*}=V_{m^*}+v_{m^*,n}$, $\Omega_{m^*}=\Omega_{m^*} \cup \{n\}$.
		\EndFor
		\While {iteration is not terminated}    \Comment{\textit{Main iteration}}
			\For {$n=1,... ,|\mathcal{N}|$}    \Comment{\textit{Insertion}}
			\State Find master $m_1$ that worker $n$ is serving. 
			\State $m_2=\arg \min_{m\in \mathcal{M}/\{m_1\}} V_{m}$.
			\State $V'_{m_1}=V_{m_1}-v_{m_1,n}$, $V'_{m_2}=V_{m_2}+v_{m_2,n}$.
			\State $V'_{m}=V_{m}, \forall m \in\mathcal{M}/\{m_1,m_2\}$.
				\If {$\min_ {m\in \mathcal{M}}V'_{m}>\min_ {m\in \mathcal{M}}V_{m}$}
				\State $\Omega_{m_1}=\Omega_{m_1}-\{n\}$, $\Omega_{m_2}=\Omega_{m_2}+\{n\}$.
				\EndIf
			\EndFor
			\For {$n_1,n_2=1,... |\mathcal{N}|$}    \Comment{\textit{Interchange}}
			\State Masters $m_1, m_2$ served by workers $n_1, n_2$, $V'_{m_1}=V_{m_1}-v_{m_1,n_1}+v_{m_1,n_2}$, and $V'_{m_2}=V_{m_2}-v_{m_2,n_2}+v_{m_2,n_1}$. 
			\If {$m_1 \neq m_2$, $n_1 \neq n_2$, $v_{m_1,n_1}+v_{m_2,n_2}<v_{m_1,n_2}+v_{m_2,n_1}$, $V'_{m_1}>V_{\text{min}}$, and $V'_{m_2}>V_{\text{min}}$}
				\State $\Omega_{m_1}=\Omega_{m_1}-\{n_1\}+\{n_2\}$.
				\State $\Omega_{m_2}=\Omega_{m_2}-\{n_2\}+\{n_1\}$.
			\EndIf
			\EndFor
			\State Randomly remove a subset of $\mathcal{N}_s$ workers, and update $V_m$ based on the current assignment.  \Comment{\textit{Exploration}}
			\While {$\mathcal{N}_s \neq \emptyset$}
			\State $\{m^*,n^*\}=\arg\max_{m\in \mathcal{M}, n\in\mathcal{N}_s } v_{m,n}$.
			\State $V_{m^*}=V_{m^*}+v_{m^*,n^*}$.
			\State $\Omega_{m^*}=\Omega_{m^*} \cup \{n^*\}$, $\mathcal{N}_s=\mathcal{N}_s-\{n^*\}$.
			\EndWhile
			\EndWhile
	\end{algorithmic}
\end{algorithm}

\begin{algorithm} [!t]
	\caption{Simple Greedy Algorithm for Dedicated Worker Assignment}
	\begin{algorithmic}[1]
		\State \textbf{Input}: $\mathcal{M}_0=\{1,2,...,M\}$, $\mathcal{N}_0=\{1,2,...,N\}$, $\Omega_m=\emptyset$, $V_m=0$, and $\{v_{m,n}\}$ according to \eqref{vmn}.
		\While {$\mathcal{M}_0 \neq \emptyset$} \Comment{\textit{Initialization}}
		\State $\{m^*,n^*\}=\arg\max_{m\in \mathcal{M}_0, n\in\mathcal{N}_0 } v_{m,n}$.
		\State $V_{m^*}=V_{m^*}+v_{m^*,n^*}$.
		\State $\Omega_m=\Omega_m \cup n^*$, $\mathcal{M}_0=\mathcal{M}_0-\{m^*\}$, $\mathcal{N}_0=\mathcal{N}_0-\{n^*\}$.
		\EndWhile
		\While {$\mathcal{N}_0 \neq \emptyset$}  \Comment{\textit{Main loop}}
		\State Find $m^*=\arg\min_{m\in \mathcal{M}}V_m$.
		\State Find $n^*=\arg\max_{ n\in\mathcal{N}_0 } v_{m^*,n}$.
		\State $V_{m^*}=V_{m^*}+v_{m^*,n^*}$.
		\State $\Omega_m=\Omega_m \cup n^*$, $\mathcal{N}_0=\mathcal{N}_0-\{n^*\}$.
		\EndWhile
	\end{algorithmic}
\end{algorithm}

Problem $\mathcal{P}5$ is a combinatorial optimization problem named \emph{max-min allocation}, which is motivated by the fair allocation of indivisible goods \cite{bryan1982scheduling, chakrabarty2009on,asadpour2010an}. Specifically, there are $M$ agents and $N$ items. Each item has a unique value for each agent, and can only be allocated to one agent. The goal is to maximize the minimum sum value of agents, by allocating items as fairly as possible. 
In our problem, \emph{each master corresponds to an agent with sum value $V_m$, and each worker $n$ can be considered as an item with value $v_{m,n}$ for master $m$}. 
The problem can be reduced to a NP-complete partitioning problem \cite{hayes2002}, when considering only $2$ agents and that each item has identical value for both agents. Therefore, problem $\mathcal{P}5$ is NP-hard.
An $O(N^\epsilon)$-approximation algorithm in time $N^{O(\frac{1}{\epsilon})} $ is proposed in \cite{chakrabarty2009on} for max-min allocation, with $\epsilon \geq \frac{9\log\log N}{\log N}$. Another polynomial-time algorithm is proposed in \cite{asadpour2010an}, guaranteeing $O(\frac{1}{M\log^3M})$ approximation to the optimum.
However, these algorithms have high computational complexity, and are difficult to implement.
We propose two low-complexity greedy algorithms as follows. 

An iterated greedy algorithm is proposed in Algorithm 1, which is inspired by \cite{peyro2010iterated}, where a similar min-max fairness problem is investigated. In the initialization phase, each worker is assigned to the master for which its value $v_{m,n}$ is the highest. The main iteration has the following three phases:

1) \emph{Insertion:} We extract each worker $n$ from the current master $m_1$, and assign it to a master $m_2\neq m_1$ with minimum sum value $V_{m_2}$. 
As shown in Lines 12-14, if the minimum sum value among masters is improved, let worker $n$ serve master $m_2$. The complexity is $O(MN)$.

2) \emph{Interchange:} We pick two workers $n_1$, $n_2$ that serve two masters $m_1$, $m_2$, and interchange their assignments. If the minimum sum $\min V_m$ is improved, and the overall system performance is improved (i.e., $v_{m_1,n_1}+v_{m_2,n_2}<v_{m_1,n_2}+v_{m_2,n_1}$), the interchange is kept. The complexity is $O(N^2)$.
Note that the insertion and interchange are repeated for multiple times within each iteration, in order to obtain a local optimum.

3) \emph{Exploration:} We randomly remove some workers from the current assignment, and allocate them in a greedy manner. This operation can be regarded as an exploration, which prevents the algorithm to be stuck in a local optimum.
When the number of iterations reach a predefined maximum, or the performance does not improve any more, the main loop is terminated. Note that the final output is the assignment obtained before the exploration phase. 

While Algorithm 1 still requires iterations to obtain a good assignment, Algorithm 2, which is inspired by the \emph{largest-value-first} algorithm in \cite{bryan1982scheduling}, is even simpler with only one round. In a homogeneous case with $v_{1,n}=\cdots=v_{M,n}$, the algorithm finds an agent $m$ with minimum sum value $V_m$, and assigns a remaining item with the largest value $v_{m,n}$. The algorithm guarantees a $\frac{4}{3}$ approximation to the optimum.
We extend the idea of the largest-value-first to the heterogeneous environment, and propose a simple greedy algorithm. As shown in Algorithm 2, in the initialization phase, we find a master without any workers assigned, and allocate an available worker that has the largest contribution for it. In the main loop, we always select master $m$ with the minimum sum value $V_m$, and allocate a remaining worker that has the maximum value $v_{m,n}$ for this master. The overall complexity of the simple greedy algorithm is $O(N^2)$.

\section{Probabilistic Worker Assignment} \label{oppor}
In this section, we solve problem $\mathcal{P}3$ for the probabilistic worker assignment, where $\mathcal{K}=[0,1]$. The key challenge is the non-convexity of constraint \eqref{cons_exp}. We observe that constraint \eqref{cons_exp} can be decomposed into the difference of convex functions, and adopt SCA method to jointly solve the worker assignment and load allocation problems.

From Lemma 1, we know that $f(l,t)$ defined in \eqref{fxt} is convex. Thus $le^{-\frac{u t}{l}}$ is also convex with respect to $l$ and $t$. 
Let $\w\triangleq\{l,k,t\}$, $g(\w)\triangleq-kl$, and $h(\w)\triangleq kle^{-\frac{u t}{l}}$. It is easy to see that the following functions are all convex:
\begin{align}
	&g^{+}(\w)\triangleq\frac{1}{2}\left(k^2+l^2\right), ~~~~~~~g^{-}(\w)\triangleq\frac{1}{2}\left(k+l\right)^2, \\
	&h^{+}(\w)\triangleq\frac{1}{2}\left(k+le^{-\frac{u t}{l}}\right)^2, ~ h^{-}(\w)\triangleq\frac{1}{2}\left(k^2+l^2e^{-\frac{2u t}{l}}\right), 
\end{align}
and we have
\begin{align}
	&g(\w)=g^{+}(\w)-g^{-}(\w), ~~~h(\w)=h^{+}(\w)-h^{-}(\w).
\end{align}

%

\begin{algorithm} [!t]
	\caption{SCA-based Probabilistic Worker Assignment and Load Allocation Algorithm}
	\begin{algorithmic}[1]
		\State \textbf{Input}: find a feasible point of  $\mathcal{P}3$, $\z_0$, set $\gamma_0=1$, $r=0$, $\alpha \in (0,1)$.
		\While {$\z_r$ is not a stationary solution} 
		\State Solve the optimal solution $\w_r$ of $\mathcal{P}(\z_r)$.
		\State $\z_{r+1}=\z_r+ \gamma_r (\w_r-\z_r)$.
		\State $\gamma_{r+1}=\gamma_r(1-\alpha \gamma_r)$, $r \leftarrow r+1$.
		\EndWhile
	\end{algorithmic}
\end{algorithm}

By linearizing the concave parts $-g^{-}(\w)$ and $-h^{-}(\w)$, given any two points $\w$, $\z$, the convex upper approximations of $g(\w)$ and $h(\w)$ can be obtained as follows \cite{scutari2017p1}:
\begin{align}
	\tilde{g}(\w,\z)\triangleq&g^{+}(\w)-g^{-}(\z)    -\nabla_{\w}g^{-}(\z)^T(\w-\z) \geq  g(\w), \\
	\tilde{h}(\w,\z)\triangleq&h^{+}(\w)-h^{-}(\z)   -\nabla_{\w}h^{-}(\z)^T(\w-\z) \geq  h(\w).  
\end{align}

Let subscript $\{m,n\}$ denote the variables, parameters and functions related to master $m$ and worker $n$, e.g., $\w_{m,n}=\{l_{m,n},k_{m,n},t\}$, $h_{m,n}(\w_{m,n})=l_{m,n}k_{m,n}e^{-\frac{u_{m,n} t}{l_{m,n}}}$; and thus,
\begin{align}
	-\mathbb{E}[X_m(t)]=\sum_{n=1}^N \left[{g}_{m,n}(\w_{m,n})+e^{u_{m,n}a_{m,n}} {h}_{m,n}(\w_{m,n}) \right]. 
\end{align}
Let $\w_m\triangleq \{\w_{m,1},...,\w_{m,N}\}$, $\z_m\triangleq\{\z_{m,1},...,\z_{m,N}\}$.
Now we can give a convex upper approximation for the left-hand side of constraint \eqref{cons_exp_simple} in the following lemma.

\begin{lemma}
	The left-hand side of constraint \eqref{cons_exp_simple} can be approximated by a convex function as follows: 
	\begin{align}
	L_m&-\mathbb{E}[X_m(t)]\leq L_m+\sum_{n=1}^N \big[ \tilde{g}_{m,n}(\w_{m,n},\z_{m,n})   \nonumber \\
	& +e^{u_{m,n}a_{m,n}} \tilde{h}_{m,n}(\w_{m,n},\z_{m,n}) \big]\triangleq \tilde{q}_m(\w_m,\z_m). \label{cons_sca_left} 
	\end{align}
\end{lemma}

Let $\z\triangleq\{\z_1, ..., \z_M\}$ be a feasible point of $\mathcal{P}3$. The convex approximation to $\mathcal{P}3$ at point $\z$, defined as $\mathcal{P}(\z)$, is given by:
\begin{subequations}
\begin{align}
	\mathcal{P}(\z): \min_{\{l_{m,n},k_{m,n},t\}} &~~~~~~t  \\
	\text{s.t.} ~~~~~& \tilde{q}_m(\w_m,\z_m)\leq 0, ~~\forall m, \label{cons_sca} \\
	&\text{Constraints } \eqref{ori_cons_sum}, \eqref{cons_l}. \nonumber
\end{align}
\end{subequations}

A probabilistic worker assignment and load allocation algorithm is proposed in Algorithm 3 based on the SCA method. 
A diminishing step-size rule is adopted with decreasing ratio $\alpha\in(0,1)$, guaranteeing the convergence of the SCA \cite{scutari2017p1}, and in line 5, $\gamma_r$ is the step-size in the $r$th iteration.
Starting from a feasible point $\z_0$ of $\mathcal{P}3$, we iteratively solve convex optimization problems $\mathcal{P}(\z_r)$, in which constraint \eqref{cons_exp_simple} is replaced by its upper convex approximation \eqref{cons_sca}. 
The iteration terminates when the solution is stationary (e.g., $\left\lVert \w_r -\z_r \right\rVert_2\leq \epsilon$), and according to Theorem 2 in \cite{scutari2017p1}, the stationary solution obtained by the SCA based algorithm is a local optimum.


\subsection{Comparison of Dedicated and Probabilistic Assignments}
We remark that the completion time of probabilistic worker assignment is a lower bound on what is achieved by dedicated worker assignment, since any feasible point of dedicated assignment is also feasible for probabilistic assignment.
However, dedicated assignment simplifies the connections between workers and masters, and requires less communication for the multicast of $\x_m$ and less storage at each worker.
Moreover, the proposed dedicated assignment algorithms have lower computational complexity and are easier to  implement.

\section{Simulation Results}  \label{sim}

In this section, we evaluate the average task completion time of the proposed dedicated and probabilistic worker assignment algorithms, in both small-scale and large-scale scenarios. In the small-scale scenario, we consider $M=2$ masters and $N=20$ workers, and three benchmarks:
1) \emph{Uncoded computing with uniform dedicated worker assignment}: each master is assigned an equal number of $\frac{N}{M}$ workers, and $\A_m$ is equally partitioned into $\frac{N}{M}$ sub-matrices without coding, each with $\frac{L_mM}{N}$ rows.
2) \emph{Coded computing with uniform dedicated worker assignment} \cite{codedhet}: each master is assigned an equal number of $\frac{N}{M}$ workers, and the load is allocated according to Theorem 1.
3) \emph{Brute-force search for dedicated worker assignment}: the oracle solution for dedicated worker assignment is obtained by searching all possible combinations, and the load is allocated according to Theorem 1.
In the large-scale scenario, we consider $M=4$ masters and $N=50$ workers, and only use the first two benchmarks, due to the high complexity of the brute-force search.

The straggling parameter $u_{m,n}$ is randomly selected within $[1, 5] ~\mathrm{ms}^{-1}$, the shift parameter is set as $a_{m,n}=\frac{1}{u_{m,n}}~\mathrm{ms}$, and $L_m=10^5$, $\forall m$ \cite{codedhet}. In Algorithm 1, we randomly remove $\frac{N}{M}$ workers for each exploration. In Algorithm 3, we set the convergence criteria as $|1-\frac{t'}{t}|<10^{-6}$, decreasing ratio $\alpha=10^{-3}$, and use CVX toolbox\footnote{http://cvxr.com/cvx/} to solve each convex approximation problem. We obtain the worker assignment and load allocation from the algorithms that minimize the approximate completion time. Then we carry out $10^5$ Monte Carlo realizations and calculate the empirical cumulative distribution function (CDF) and the average of task completion time. 

\begin{figure}[!t]
	\centering
	\includegraphics[width=0.65\textwidth]{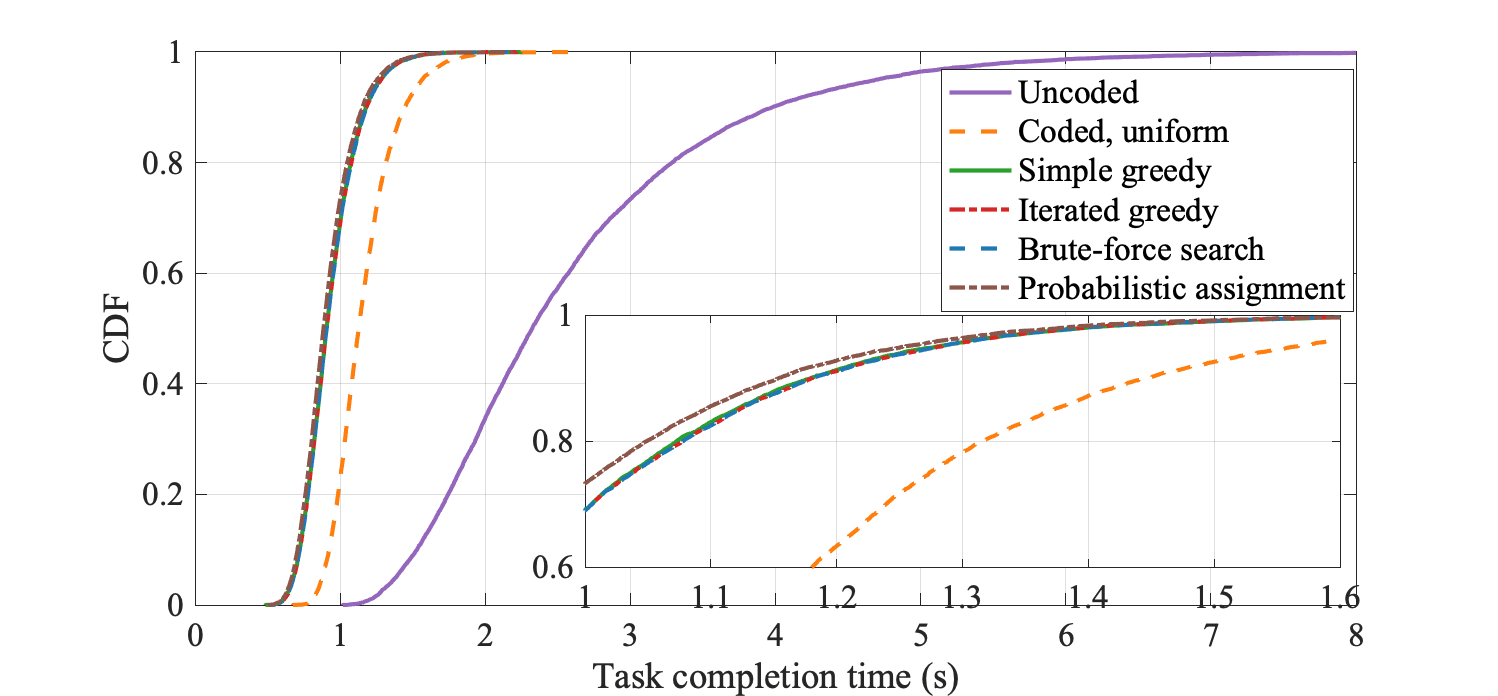}
	\caption{The CDF of task completion time achieved by different worker assignment and load allocation algorithms with $2$ masters and $20$ workers.}	\label{bar_m2n20}
\end{figure}

\begin{figure}[!t]
	\centering
	\includegraphics[width=0.65\textwidth]{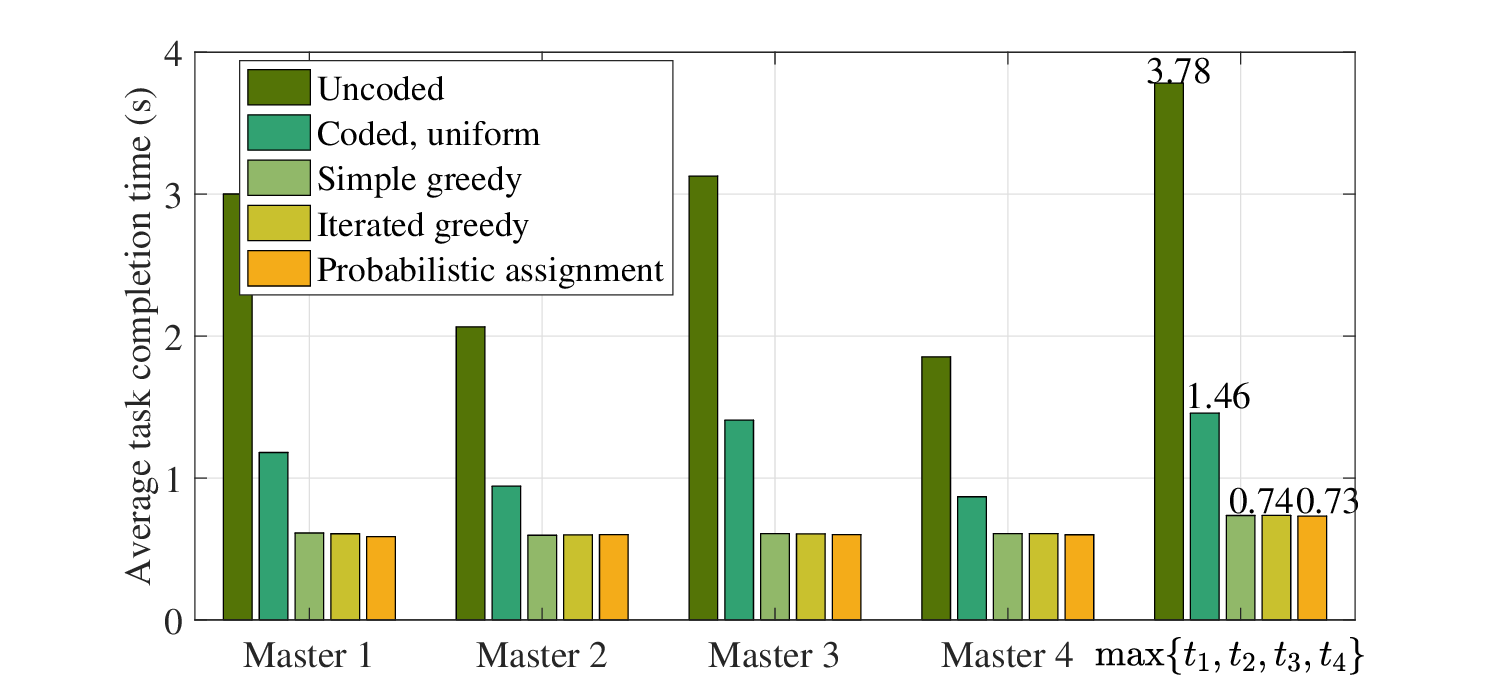}
	\caption{Average task completion time achieved by different worker assignment and load allocation algorithms with $4$ masters and $50$ workers.}	\label{bar_m4n50}
\end{figure}

\begin{figure}[!t]
	\centering
	\includegraphics[width=0.65\textwidth]{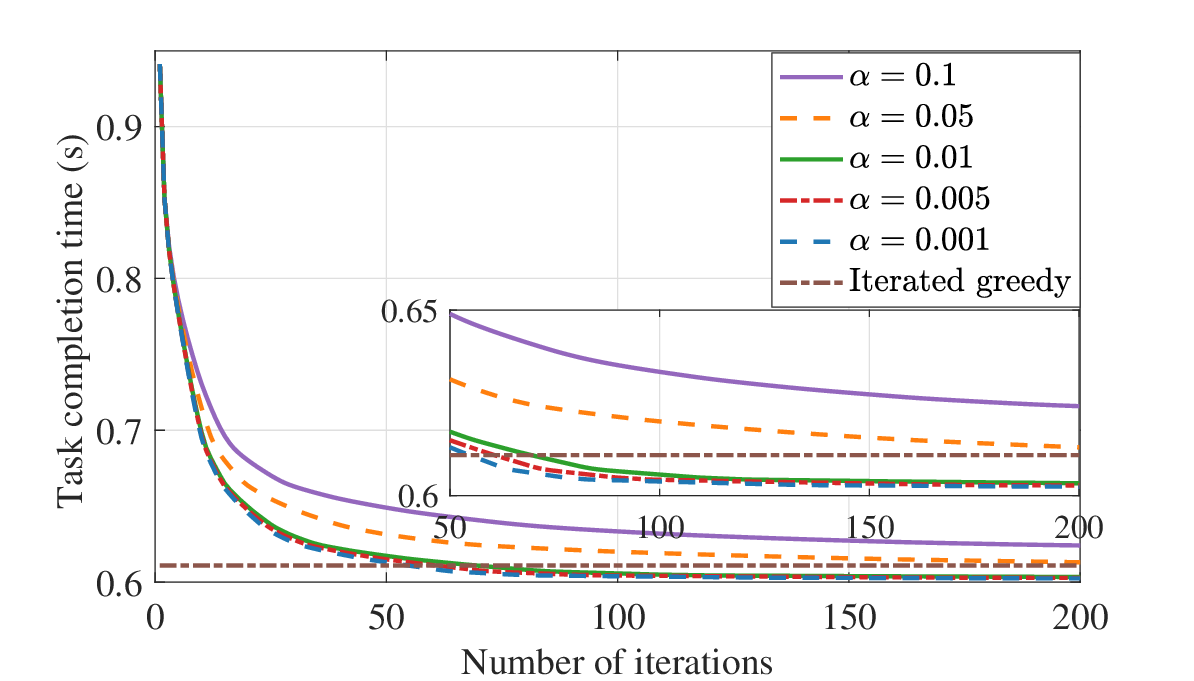}
	\caption{Convergence of the SCA-based probabilistic worker assignment algorithm with $4$ masters and $50$ workers.}
	\label{convergence}
\end{figure}

Fig. \ref{bar_m2n20} shows the CDFs of the task completion time.
The proposed greedy dedicated assignment and SCA-based probabilistic assignment algorithms outperform the uncoded and coded benchmarks with uniform assignment of dedicated workers.
The CDFs achieved by iterated and greedy algorithms are very close, and both performances are close to the optimal brute-force search algorithm. Specifically, when the successful probability $\rho_s=0.98$, the three dedicated assignment algorithms all achieve task completion time $1.40\mathrm{s}$.
Probabilistic assignment further outperforms the dedicated assignment, which is consistent with the fact that it is a lower bound for dedicated assignment. When $\rho_s=0.98$, probabilistic assignment achieves task completion time $1.38\mathrm{s}$.

Fig. \ref{bar_m4n50} compares the average task completion time achieved by the proposed algorithms and benchmarks.
The first four groups of bars show the average time each master needs to finish its own task using different algorithms. The fifth group of bars show the average task completion time of the system, which is what we aim to minimize, obtained by averaging the maximum time of each realization.
From the fifth group of bars, we can see that all the proposed algorithms reduce the delay performance by more than $80\%$ over uncoded benchmark, and more than $49\%$ over coded benchmark.
The performance gain is mainly achieved by taking into account the heterogeneity of the system.
From the first four groups of bars, we can see that the average delay of each master achieved by our proposed algorithms are very close, indicating that the workers and loads are assigned in a balanced manner.

In Fig. \ref{convergence}, the impact of the decreasing ratio $\alpha$ on the convergence of SCA-based probabilistic assignment algorithm is evaluated, in the scenario with $4$ masters and $50$ workers. The decreasing ratio $\alpha$ decides the step-size $\gamma_r$, and thus the convergence rate of the SCA algorithm.
We can see that by choosing a proper $\alpha$, the proposed SCA algorithm can converge after $100$ iterations, and outperforms the iterated greedy algorithm for dedicated worker assignment.

\section{Conclusions} \label{con}
We have considered a joint worker assignment and load allocation problem in a distributed computing system with heterogeneous computing servers, i.e., workers, and multiple master nodes competing for these workers. MDS coding has been adopted by the masters to mitigate the straggler effect, and both dedicated and probabilistic assignment algorithms have been proposed, in order to minimize the average task completion time.
Simulation results show that the proposed algorithms can reduce the task completion time by $80\%$ compared to uncoded task assignment, and $49\%$ over an unbalanced coded scheme.
While probabilistic assignment is more general, we have observed through simulations that the two have similar delay performances.
We have noted that dedicated assignment has lower computational complexity and lower communication and storage requirements, beneficial for practical implementations.
As future work, we plan to take communication delay into consideration, and develop decentralized algorithms.



\appendices{}

\section{Proof of Lemma 1} \label{proof1}
It is easy to see that \eqref{p3_obj} and \eqref{p3_cons2} are convex objective and constraints, respectively.
Let 
\begin{align}
	f(x,t)=-x \left( 1-e^{-\frac{u}{x}(t-ax)}\right),  \label{fxt}
\end{align}
with variables $x>0$, $~t\geq ax$, and parameters $u>0$, $a>0$.
The Hessian matrix of $f(x,t)$ is:
\begin{align}
\H= \left[
\begin{matrix}
\frac{\partial^2 f }{\partial x^2 }  &\frac{\partial^2 f }{\partial x \partial t }\\
\frac{\partial^2 f }{\partial t \partial x} &\frac{\partial^2 f }{\partial  t^2 }
\end{matrix}
\right]=e^{-\frac{u}{x}(t-ax)} \left[
\begin{matrix}
\frac{u^2t^2}{x^3}  &-\frac{u^2t}{x^2}\\
-\frac{u^2t}{x^2}&\frac{u^2}{x}
\end{matrix}
\right].
\end{align}
The eigenvalues of $\H$ are $0$ and $\frac{u^2(x^2+t^2)}{x^3}>0$. Thus $\H \succeq 0$, and $f(x,t)$ is convex.
Let $u=u_{m,n}$ and $a=a_{m,n}$, $-\mathbb{E}[X_{m,n}(t_m)]=f(l_{m,n},t_m)$ is convex.
Constraint \eqref{p3_cons1} 
is the summation of convex functions, and hence convex.
Therefore, $\mathcal{P}4$ is a convex optimization problem.

\section{Proof of Theorem 1} \label{proof2}
By jointly considering \eqref{partial_t} and \eqref{kkt_2}, we can get $\lambda_m^*>0$. Then,
substituting \eqref{partial_l} into $\frac{\partial \mathcal{L}}{\partial l_{m,n}^*}=0$, we have:
\begin{align}
-\left(1+\frac{t_m^*u_{m,n}}{l_{m,n}^*}\right)e^{-\left(1+\frac{t_m^*u_{m,n}}{l_{m,n}^*}\right)}=-e^{-u_{m,n}a_{m,n}-1}, 
\end{align}
\begin{align} \label{opt_t_over_l}
	\frac{t^*_m}{l^*_{m,n}}=\frac{-\mathcal{W}_{-1}(-e^{-u_{m,n}a_{m,n}-1} )-1}{u_{m,n}}=\phi_{m,n}.
\end{align}
Substituting \eqref{opt_t_over_l} into \eqref{kkt_3}, we have
\begin{align}
&L_m-\sum_{n\in \Omega_m} \frac{t^*_m}{\phi_{m,n}}\left(1-\frac{1}{1+u_{m,n}\phi_{m,n}}\right)=0.
\end{align}
Thus, $t_m^* $ and $l^*_{m,n}$ can be derived as in Theorem 1.


\end{document}